\documentclass{llncs}
\usepackage[utf8]{inputenc}
\usepackage{capt-of}
\usepackage{graphicx} 
\usepackage{latexsym}
\usepackage{url}
\usepackage{color}
\usepackage{multirow}
\usepackage{xspace}
\usepackage{multicol}
\usepackage{alltt}
\usepackage{amssymb}
\usepackage{amsmath}
\usepackage{ifthen}
\usepackage{bussproofs}
\usepackage{wrapfig}
\usepackage{microtype}

\newcommand{\hkb}{\variable{\qclauset_t}}
\newcommand{\eVar}{\variable{e}}
\newcommand{\uVar}{\variable{x}}
\newcommand{\variable}[1]{\ensuremath{#1}\xspace}
\newcommand{\eVarNeg}{\variable{\overline{\eVar}}}
\newcommand{\uVarNeg}{\variable{\overline{\uVar}}}
\newcommand{\clause}{\variable{C}}

\newcommand{\rareqs}{RAReQS\xspace}
\newcommand{\ghostq}{GhostQ\xspace}
\newcommand{\caqe}{CAQE\xspace}
\newcommand{\qesto}{QESTO\xspace}
\newcommand{\depqbf}{DepQBF\xspace}
\newcommand{\qell}{QELL\xspace}
\newcommand{\qellc}{QELL-c\xspace}
\newcommand{\qellnc}{QELL-nc\xspace}
\newcommand{\depqbfDynBloqqerTFTT}{DQ-BAT\xspace}
\newcommand{\depqbfDynTFTT}{DQ-AT\xspace}
\newcommand{\depqbfOnlyDynBloqqer}{DQ-B\xspace}
\newcommand{\depqbfOnlyDynQBCE}{DQ\xspace}
\newcommand{\depqbfOnlyDynTF}{DQ-A\xspace}
\newcommand{\depqbfOnlyDynTT}{DQ-T\xspace}
\newcommand{\depqbfNoQBCE}{DQ-nQ\xspace}
\newcommand{\depqbfNoQBCEDynBloqqerTFTT}{DQ-nQBAT\xspace}
\newcommand{\depqbfNoQBCETFTT}{DQ-nQAT\xspace}
\newcommand{\depqbfNoQBCETT}{DQ-nQT\xspace}

\newcommand{\true}{\textsf{T}\xspace}
\newcommand{\false}{\textsf{F}\xspace}

\newcommand{\bloqqer}{Bloqqer\xspace}
\newcommand{\prefix}{\Pi}
\newcommand{\clauset}{\psi}
\newcommand{\qclauset}{\phi}
\newcommand{\dec}[1]{\mathsf{dec}(#1)}
\newcommand{\der}[1]{\mathsf{der}(#1)}
\newcommand{\lit}[1]{\mathsf{var}(#1)}

\newcommand{\quant}[2]{\mathsf{Q}(#1,#2)}

\newcommand{\UR}{\mathit{UR}}
\newcommand{\ER}{\mathit{ER}}
\newcommand{\satequiv}{\equiv_{\mathit{sat}}}

\newcommand{\qrescalc}{QRES\xspace}

\newcommand{\eabs}[1]{\mathit{Abs_{\exists}}(#1)}
\newcommand{\aabs}[1]{\mathit{Abs_{\forall}}(#1)}

\begin{document}

\author{Florian Lonsing\inst{1} \and
 Uwe Egly\inst{1}
\and Martina Seidl\inst{2}}

\institute{Knowledge-Based Systems Group, Vienna University of Technology, Vienna, Austria \and 
Institute for Formal Models and Verification, JKU, Linz, Austria
}

\title{Q-Resolution with Generalized Axioms\thanks{Supported by the Austrian Science Fund (FWF)
    under grants S11408-N23 and S11409-N23. This article will appear in the
\textbf{proceedings} of the \emph{19th International Conference on Theory and
Applications of Satisfiability Testing (SAT)}, LNCS, Springer, 2016.}}
\maketitle
\begin{abstract}
Q-resolution is a proof system for quantified Boolean formulas (QBFs) in
prenex conjunctive normal form (PCNF) which underlies search-based QBF solvers
with clause and cube learning (QCDCL).
With the aim to derive and learn stronger 
clauses and cubes earlier in the search, 
we generalize the axioms of the Q-resolution calculus
resulting in an exponentially more powerful proof system. 
The generalized axioms introduce an interface of Q-resolution to any 
other QBF proof system allowing for the direct combination 
of orthogonal solving techniques. 
We implemented a variant of the Q-resolution calculus with
generalized axioms in the QBF solver DepQBF. As two case studies, 
we apply integrated SAT
solving and resource-bounded QBF preprocessing during the search to
heuristically detect potential axiom applications. Experiments with
application benchmarks indicate a substantial performance 
improvement.
\end{abstract}


\section{Introduction}

In the same way as SAT, the decision problem of propositional logic, is the
archetypical problem complete for the complexity class NP, 
QSAT, the decision problem of \emph{quantified Boolean formulas (QBF)},
is the archetypical problem complete for the complexity class PSPACE. 
The fact that many important practical reasoning, verification, 
and synthesis problems 
fall into the latter complexity class (cf.~\cite{DBLP:journals/jsat/BenedettiM08} for an overview)
strongly motivates the quest for efficient QBF solvers. 

As the languages of propositional logic and QBF only marginally differ 
from a syntactical point of view, namely the quantifiers, it is 
a natural approach to take inspiration from SAT solving and 
lift powerful SAT techniques to QSAT.
Motivated by the success of \emph{conflict-driven clause learning} 
(CDCL) in SAT solving~\cite{DBLP:series/faia/SilvaLM09}, 
a generalized version of CDCL called 
\emph{conflict/solution-driven clause/cube learning} 
(often abbreviated by QCDCL) is applied in QSAT 
solving~\cite{DBLP:series/faia/GiunchigliaMN09}. 
Given a propositional 
formula in conjunctive normal form (CNF), a CDCL-based SAT solver
enriches the original CNF with clauses---already found and justified 
conflicts---which force the solver into a different area of the 
search space until either a model, i.e., a satisfying variable assignment, 
is found or until the CNF is proven to be unsatisfiable.  
If a QBF in prenex conjunctive normal form (PCNF) is unsatisfiable
then QCDCL works similar, apart from technical details. 
In the case of satisfiability, however, 
it is not sufficient to find one assignment satisfying 
the formula. To respect the semantics of universal quantification, 
QBF models have to be described either by assignment trees or 
by Skolem functions. Hence, a QBF solver may not abort the search 
if a satisfying assignment is found. Dual to clause learning, a 
cube (a conjunction of literals) is learned and the search is resumed. 
QCDCL is
implemented in several state-of-the-art 
QBF \mbox{solvers~\cite{DBLP:journals/jair/GiunchigliaNT06,DBLP:journals/iandc/BuningKF95,DBLP:conf/tableaux/Letz02,DBLP:conf/iccad/ZhangM02}.}

Apart from QCDCL, orthogonal approaches to QBF solving have been developed.  
QBF competitions like the QBF Galleries 2013~\cite{Lonsing201692} 
and 2014~\cite{gallery14} revealed 
the power of 
\emph{expansion-based approaches}~\cite{DBLP:conf/fmcad/AyariB02,DBLP:conf/sat/Biere04a,Janota20161},
which are based on a different 
proof system than search-based solving with QCDCL.
We refer to related work~\cite{DBLP:conf/mfcs/BeyersdorffCJ14,beyersdorff_et_al:LIPIcs:2015:4905} for an overview of QBF 
proof systems.
QCDCL relies on 
Q-resolution~\cite{DBLP:journals/iandc/BuningKF95}. 
Traditionally, Q-resolution calculi\footnote{Note that there are 
different variants of Q-resolution, e.g., long-distance resolution~\cite{DBLP:conf/iccad/ZhangM02}, QU-resolution~\cite{DBLP:conf/cp/Gelder12}, etc.~\cite{DBLP:conf/sat/BalabanovWJ14,beyersdorff_et_al:LIPIcs:2015:4905}.}
offer two kinds of axioms with limited deductive power: 
(i) the \emph{clause axiom} stating that any clause in the  
CNF part of a QBF  can be immediately derived  
and (ii) the \emph{cube axiom} allowing to derive cubes which are 
propositional implicants of the CNF. 
In previous work~\cite{DBLP:conf/lpar/LonsingBBES15}, we generalized the cube axiom 
such that quantified 
blocked clause elimination (QBCE)~\cite{DBLP:conf/cade/BiereLS11}, a 
clause elimination procedure for  
preprocessing, could be tightly integrated in QCDCL 
for learning smaller cubes earlier in the search. 

To overcome the restrictions of the traditional axioms of Q-resolution, we
extend previous work~\cite{DBLP:conf/lpar/LonsingBBES15} on the cube axiom and present more
powerful clause axioms. We generalize the traditional clause and cube 
axioms such that their application relies on checking the satisfiability of
the PCNF under the current assignment in QCDCL. This way, the axioms can be
applied earlier in the search. Further, they provide a
framework to combine Q-resolution with any other (complete or incomplete) QBF
proof system. We implemented the generalized axioms in the QCDCL solver
DepQBF. As a case study, we integrated bounded expansion and SAT-based
abstraction~\cite{DBLP:conf/aaai/CadoliGS98} in QCDCL as incomplete QBF
solving techniques to detect potential axiom applications. Experimental
results indicate a substantial performance increase, particularly on
application benchmarks.

This paper is structured as follows. In Sections~\ref{sec:prelim} and~\ref{sec:qcdcl},
we introduce preliminaries and recapitulate search-based QBF solving with QCDCL and 
traditional Q-resolution. Then we 
generalize the axioms of Q-resolution in Section~\ref{sec:gen:axioms}
allowing for the integration of other proof systems. 
In Section~\ref{sec:abs:axiom} we integrate SAT-based 
abstraction into QCDCL. 
Implementation and evaluation are  discussed in Section~\ref{sec:exp}. We conclude with a summary and an outlook to future work in 
Section~\ref{sec:concl}.


\section{Preliminaries}
\label{sec:prelim}

We introduce the concepts and terminology used in 
the rest of the paper. A \emph{literal} is a variable $x$ or 
its negation $\bar x$. By $\bar l$ we denote the negation of 
literal $l$ and $\lit{l} := x$ if $l = x$ or 
$l = \bar x$. A disjunction, resp.\ conjunction, of literals 
is called \emph{clause}, resp.\ \emph{cube}.  
A propositional formula 
in \emph{conjunctive normal form} (CNF) is a conjunction of 
clauses. 
If convenient, 
we interpret a CNF as a set of clauses, and clauses and cubes as sets of literals. 
A QBF in \emph{prenex conjunctive normal form} (PCNF) 
has the form $\prefix.\clauset$ with prefix $\prefix := Q_1X_1 \ldots Q_nX_n$ and matrix $\clauset$, 
where $\clauset$ is a propositional CNF over the variables 
defined in $\prefix$. 
The variable sets $X_i$ are pairwise disjoint and for
$Q_i \in \{\forall, \exists\}$, $Q_i \not= Q_{i+1}$.
We define $\lit{\prefix} := X_1 \cup \ldots \cup X_n$. 
 The quantifier
$\quant{\prefix}{l}$ of  a literal $l$ is  $Q_i$ if
$\lit{l} \in X_i$. If $\quant{\prefix}{l} = Q_i$
and $\quant{\prefix}{k} = Q_j$,
then $l \leq_\prefix k$ iff $i \leq j$.
For a clause or cube $C$, $\lit{C} := \{\lit{l} \mid l \in C\}$
and for CNF $\clauset$, $\lit{\clauset} := \{ \lit{l} \mid l \in C, 
C \in \clauset \}$.

An \emph{assignment} $A$ is a mapping from the variables $\lit{\prefix}$ 
of a QBF $\prefix.\clauset$ to truth values \emph{true} and \emph{false}. We
represent $A$ 
as a set of literals $A = \{l_1,\ldots,l_n\}$ with $\{\lit{l_i} \mid l_i \in A\} \subseteq 
\lit{\prefix}$ such that if a variable $x$ is
assigned \emph{true} then $l_i \in A$ and $l_i = x$, and if $x$ is
assigned \emph{false} then $l_i \in A$ and $l_i = \bar x$. 
Further, for any $l_i, l_j \in A$ with $i \not= j$,  
$\lit{l_i} \not= \lit{l_j}$. An assignment $A$ is \emph{partial} if it does
not map every variable in $\lit{\prefix}$ to a truth value, i.e., $\{\lit{l_i} \mid l_i \in A\} \subset 
\lit{\prefix}$.
A QBF $\qclauset$ \emph{under assignment} $A$, 
written as 
$\qclauset[A]$, is the QBF
obtained from $\qclauset$ in which 
for all $l \in A$, all clauses containing $l$ are removed, 
all occurrences of $\bar l$ are deleted, and $\lit{l}$ is 
removed from the prefix.
If the matrix of $\qclauset[A]$ is empty, then 
the matrix is satisfied by $A$ and 
$A$ is a \emph{satisfying assignment} (written as $\qclauset[A] = \true$).
If the matrix of  $\qclauset[A]$ contains the empty clause, 
then the matrix is falsified by $A$ and 
$A$ is a \emph{falsifying assignment} 
(written as  $\qclauset[A] = \false$). 
A QBF $\prefix.\clauset$ 
with $Q_1 = \exists$ (resp.\ $Q_1 = \forall$) 
is satisfiable iff
$\prefix.\clauset[\{x\}]$ or (resp.\ and)
$\prefix.\clauset[\{\bar x\}]$ is satisfiable where
$x \in X_1$.
Two QBFs $\qclauset$ and $\qclauset'$ are \emph{satisfiability-equivalent},
written as $\qclauset \satequiv \qclauset'$, iff
$\qclauset$ is satisfiable whenever $\qclauset'$ is
satisfiable. Two propositional CNFs $\clauset$ and $\clauset'$ are
\emph{logically equivalent}, written as $\clauset \equiv \clauset'$, 
iff they
have the same set of propositional models, i.e., 
satisfying assignments. 
Two simplification rules preserving 
satisfiability equivalence are \emph{unit} and \emph{pure literal detection}. 
If a QBF $\qclauset$ contains a unit clause 
$C = (l)$, where $\quant{\prefix}{l} = \exists$, then 
$\qclauset \satequiv \qclauset[\{l\}]$. If a literal is 
pure in QBF $\qclauset$, i.e., $\qclauset$ contains $l$ but not $\bar l$, 
then $\qclauset \satequiv \qclauset[\{l\}]$ if $\quant{\prefix}{l} = \exists$ 
and $\qclauset \satequiv \qclauset[\{\bar l\}]$ otherwise.


\section{QCDCL-Based QBF Solving}
\label{sec:qcdcl}

Figure~\ref{fig:qcdcl} shows an abstract workflow of traditional search-based
QBF solving with
QCDCL~\cite{DBLP:journals/jair/GiunchigliaNT06,DBLP:journals/iandc/BuningKF95,DBLP:conf/tableaux/Letz02,DBLP:conf/iccad/ZhangM02}.
Given a PCNF $\phi$, assignments $A$ are successively generated (box in top
left corner of Fig.~\ref{fig:qcdcl}). In general,
variables must be assigned in the ordering of the
quantifier prefix. Variables may either be assigned tentatively as
\emph{decisions} or by a QBF-specific variant of \emph{Boolean constraint
propagation (QBCP)}. QBCP consists of unit and pure literal
detection. Assignments of variables carried out in QBCP do not have to follow
the prefix ordering. We formalize the assignments generated during a run of
QCDCL as follows.

\begin{figure}[t]
\centering
\includegraphics{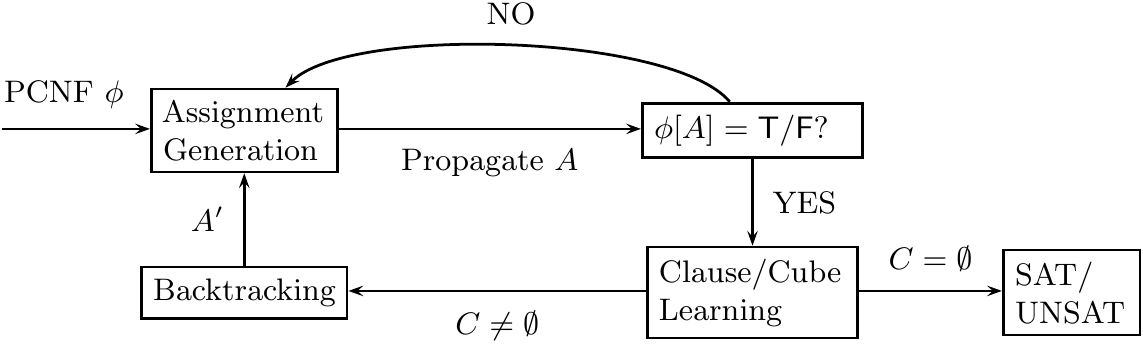}
\caption{Abstract workflow of QCDCL with traditional Q-resolution axioms.}
\label{fig:qcdcl}
\end{figure}

\begin{definition}[QCDCL Assignment]
\label{def:assign} Given a QBF $\qclauset = \prefix.\clauset$. Let assignment
$A = A' \cup A''$ where $A'$ are variables assigned as decisions and $A''$ are
variables assigned by unit/pure literal detection.  $A$ is a \emph{QCDCL
assignment} if (1) for a maximal $l \in A'$ with $\forall l' \in A' : l'
\leq_\prefix l$ it holds that $\forall x \in \lit{\prefix} <_\prefix l : x \in
\lit{A}$ and (2) all $l \in A''$ are unit/pure in $\qclauset[A']$ after
applying QBCP until completion.
\end{definition}

QCDCL generates only QCDCL assignments by
Definition~\ref{def:assign}. Assignment generation by decisions and QBCP
continues until the current assignment~$A$ is either falsifying or satisfying
by checking whether $\qclauset[A]
= \false$ or $\qclauset[A] = \true$ (box in top
right corner of Fig.~\ref{fig:qcdcl}). In these cases, a new
\emph{learned clause} or \emph{learned cube} is derived in a \emph{learning} phase,
which is based on the
\emph{Q-resolution calculus}.

\begin{definition}[Q-Resolution Calculus] \label{def_qres_calculus}
Let $\qclauset = \prefix.\clauset$ be a PCNF. The rules of the \emph{Q-resolution 
calculus (\qrescalc)} are as follows. 
\begin{align}\tag{$\mathit{res}$}
\AxiomC{$C_1 \cup \{p\}$}
\AxiomC{$C_2 \cup \{\bar p\}$}
\BinaryInfC{$C_1 \cup C_2$}
\label{rule_res}
\DisplayProof
\quad
\begin{minipage}{0.55\textwidth}
if for all $x \in \prefix\colon \{x, \bar x\} \not \subseteq (C_1 \cup C_2)$, 
\\
$\bar p \not \in C_1$, $p
\not \in C_2$, and either \\(1) $C_1$,$C_2$ are
clauses and $\quant{\prefix}{p} = \exists$ or \\ (2) $C_1$,$C_2$ are cubes and $\quant{\prefix}{p} = \forall$
\end{minipage} 
\end{align}

\vspace{-0.25cm}

\begin{align}\tag{$\mathit{red}$}
\AxiomC{$C \cup \{l\}$}
\UnaryInfC{C}
\label{rule_red}
\DisplayProof
\quad
\begin{minipage}{0.75\textwidth}
if for all $x \in \prefix\colon \{x, \bar x\} \not \subseteq (C \cup \{l\})$ and either  
\\
(1) $C$ is a
clause, $\quant{\prefix}{l} = \forall$, \\
\hspace*{0.5cm} 
$l' <_\prefix l$ 
for all $l' \in C$ with $\quant{\prefix}{l'} = \exists$
or \\ 
(2) $C$ is a
cube, $\quant{\prefix}{l} = \exists$, \\
\hspace*{0.5cm} 
$l' <_\prefix l$ 
for all $l' \in C$ with $\quant{\prefix}{l'} = \forall$
\end{minipage} 
\end{align}

\vspace{-0.25cm}

\begin{align}\tag{$\operatorname{\emph{cl-init}}$}
\AxiomC{\phantom{A}}
\UnaryInfC{C}
\label{rule_cl_init}
\DisplayProof
\quad
\begin{minipage}{0.75\textwidth}
if for all $x \in \prefix\colon \{x, \bar x\} \not \subseteq C$,  
$C$ is a clause and $C \in \clauset$
\end{minipage} 
\end{align}

\vspace{-0.35cm}

\begin{align}\tag{$\operatorname{\emph{cu-init}}$}
\AxiomC{\phantom{A}}
\UnaryInfC{C}
\label{rule_cu_init}
\DisplayProof
\quad
\begin{minipage}{0.75\textwidth}
$A$ is a QCDCL assignment,\\ 
$\qclauset[A] = \textnormal{\true}$, \\
and $C = (\bigwedge_{l \in A} l)$ is a cube
\end{minipage} 
\end{align}
\end{definition}

\qrescalc is a proof system which underlies QCDCL. Rule~\ref{rule_cl_init} is
an axiom to derive clauses which are already part of the given PCNF
$\qclauset$. In practice, the clause $C$ selected by axiom~\ref{rule_cl_init} is
falsified under the current QCDCL assignment. Axiom~\ref{rule_cu_init} allows to derive cubes based on a
QCDCL assignment $A$ which satisfies all the clauses of the matrix
$\clauset$ of $\qclauset = \prefix.\clauset$ (i.e., $\qclauset[A] =
\true$). A cube $C$ derived by axiom~\ref{rule_cu_init} is an
\emph{implicant} of $\clauset$, i.e., the implication $C \Rightarrow \clauset$
is \nolinebreak valid.

The \emph{resolution} and \emph{reduction} rules~\ref{rule_res}
and~\ref{rule_red}, respectively, are applied either to clauses or
cubes. Rule~\ref{rule_red} is called \emph{universal (existential) reduction}
when applied to clauses (cubes).  We write $\UR(C)$ ($\ER(C)$) to denote the
clause (cube) resulting from universal (existential) reduction of clause
(cube) $C$. The PCNF $\UR(\qclauset)$ is obtained by universal reduction of
all clauses in the PCNF $\qclauset$.

Q-resolution of clauses~\cite{DBLP:journals/iandc/BuningKF95} generalizes
propositional resolution, which consists of rules~\ref{rule_cl_init} and
\ref{rule_res}, by the reduction rule~\ref{rule_red}. Q-resolution of cubes
was introduced for \emph{cube
learning}~\cite{DBLP:journals/jair/GiunchigliaNT06,DBLP:conf/tableaux/Letz02,DBLP:conf/iccad/ZhangM02},
the dual variant of clause learning.

\qrescalc is sound and refutationally complete for
PCNFs~\cite{DBLP:journals/jair/GiunchigliaNT06,DBLP:journals/iandc/BuningKF95,DBLP:conf/tableaux/Letz02,DBLP:conf/iccad/ZhangM02}. The
empty clause (cube) is derivable from a PCNF $\qclauset$ in \qrescalc if and
only if $\qclauset$ is unsatisfiable (satisfiable). A derivation of the empty
clause (cube) from $\qclauset$ is a \emph{clause (cube) resolution proof} of
$\qclauset$.

In QCDCL, the rules of \qrescalc are applied to derive new learned clauses or
cubes. A learned clause (cube) $C$ is added conjunctively (disjunctively) to
the PCNF $\qclauset = \prefix.\clauset$ to obtain $\prefix.(\clauset \wedge
C)$ ($\prefix.(\clauset \vee C)$). After $C$ has been added, certain
assignments in the current assignment $A$ are retracted during
\emph{backtracking}, resulting in assignment $A'$ ($C \not = \emptyset$ in Fig.~\ref{fig:qcdcl}). Assignment generation based
on $A'$ continues, where learned clauses and cubes participate in QBCP.
Typically, only \emph{asserting} learned clauses and cubes are generated in
QCDCL. A clause (cube) $C$ is asserting if $\UR(C)$ ($\ER(C)$) is unit under $A'$ after
backtracking. QCDCL terminates if and only if the empty clause or cube is
\nolinebreak learned ($C = \emptyset$ in Fig.~\ref{fig:qcdcl}).

\begin{example}[\cite{DBLP:conf/lpar/LonsingBBES15}]
\label{ex_bad_cube_proof}
Given a  PCNF $\qclauset$ with prefix $\exists z, \! z' \forall u \exists
y$ and matrix $\clauset$:

\begin{minipage}{0.45\textwidth}
\flushleft
$
\begin{array}{lll}
\clauset & := & (u \vee \bar y) \wedge (\bar u \vee y) \wedge \mbox{} \\
 & & (z \vee u \vee \bar y) \wedge (z' \vee \bar u \vee y) \wedge \mbox{} \\
& & (\bar z \vee \bar u \vee \bar y) \wedge (\bar z' \vee u \vee y) \\
\end{array}
$
\end{minipage}
\begin{minipage}{0.5\textwidth}
\flushright
\AxiomC{$\phantom{A}$}
\UnaryInfC{$(\bar z \wedge \bar z' \wedge \bar u \wedge \bar y)$}
\UnaryInfC{$(\bar z \wedge \bar z' \wedge \bar u)$}
\AxiomC{$\phantom{B}$}
\UnaryInfC{$(\bar z \wedge \bar z' \wedge u \wedge y)$}
\UnaryInfC{$(\bar z \wedge \bar z' \wedge u)$}
\BinaryInfC{$\bar z \wedge \bar z'$}
\UnaryInfC{$\emptyset$}
\DisplayProof
\medskip
\end{minipage}

\noindent Let $A_1 := \{\bar z, \bar z', \bar u, \bar y\}$ and $A_2 := \{\bar
z, \bar z', u, y\}$ be satisfying QCDCL assignments to be used for
applications of axiom~\ref{rule_cu_init}.  A derivation of the empty cube by
rules~\ref{rule_cu_init}, \ref{rule_red}, \ref{rule_res}, and~\ref{rule_red}
(from top to bottom) is shown on the right.  \hfill $\Diamond$
\end{example}


\section{Generalizing the Axioms of \qrescalc}
\label{sec:gen:axioms}

The axioms~\ref{rule_cl_init} and~\ref{rule_cu_init} of \qrescalc have limited
deductive power. Any clause derived by~\ref{rule_cl_init} already appears in
the matrix $\clauset$ of the PCNF $\qclauset = \prefix.\clauset$. Any cube
derived by~\ref{rule_cu_init} is an implicant of $\clauset$.

To overcome these limitations, we equip \qrescalc with two additional
axioms---one to derive clauses and one to derive cubes---which
generalize~\ref{rule_cl_init} and~\ref{rule_cu_init}.  \emph{Generalized model
generation (GMG)}~\cite{DBLP:conf/lpar/LonsingBBES15} was presented as a new
axiom to derive learned cubes. The combination of \qrescalc with GMG is
stronger than \qrescalc with~\ref{rule_cu_init} in terms of the sizes of cube
resolution proofs it is able to produce. In the following, we formulate a
generalized clause axiom which we combine with \qrescalc in addition to
GMG. Thereby, we obtain a variant of \qrescalc which is stronger than
traditional \qrescalc also in terms of sizes of clause resolution proofs.

\begin{figure}[t]
\centering
\includegraphics{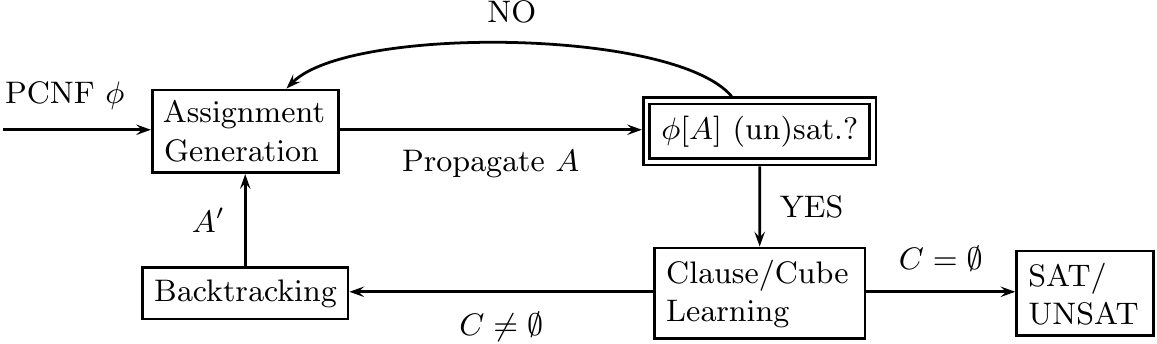}
\caption{Abstract workflow of QCDCL with generalized Q-resolution axioms.}
\label{fig:qcdcl:gen:axioms}
\end{figure}

Figure~\ref{fig:qcdcl:gen:axioms} shows an abstract workflow of search-based
QBF solving with QCDCL relying on \qrescalc \emph{with} generalized axioms. 
This workflow is the same as in Fig.~\ref{fig:qcdcl} except for applications of
axioms (box in top right corner). The
generalized axioms are applied if the PCNF $\qclauset[A]$ under a QCDCL assignment $A$ is \emph{(un)satisfiable}. This is in contrast to the
more restricted conditions $\qclauset[A] = \true$ or $\qclauset[A] = \false$
in Fig.~\ref{fig:qcdcl}. We show that the generalized axioms 
allow to combine \emph{any} sound (but maybe incomplete)
QBF solving technique with QCDCL based on \qrescalc.  First, we define \emph{QCDCL clauses} and recapitulate \emph{QCDCL
cubes}~\cite{DBLP:conf/lpar/LonsingBBES15}.
\begin{definition}[QCDCL Clause/Cube]
\label{def:assigncube}
Given a QBF $\qclauset = \prefix.\clauset$.
The \emph{QCDCL clause} $C$ of 
QCDCL assignment
$A$ is defined by $C = (\bigvee_{l
\in A} \bar l)$. 
The \emph{QCDCL cube} $C$ of 
QCDCL assignment
$A$ is defined by $C = (\bigwedge_{l
\in A} l)$.
\end{definition}
By Definition~\ref{def:assign}, a QCDCL clause or cube cannot
contain complementary literals $x$ and $\bar x$ of some variable $x$. 
According to QCDCL assignments, we split QCDCL clauses and cubes into 
decision literals and literals assigned by unit and pure literal detection.
Let $C$ be a QCDCL clause or QCDCL cube. 
Then $C = C' \cup C''$ where $C'$ is the maximal subset of $C$
such that $X_1 \cup \ldots \cup X_{i-1} \subset \lit{C'}$ 
and $C' \cap X_{i} \not= \emptyset$. The literals in $C'$ are 
the first $|C'|$ consecutive variables of $\prefix$ 
which are assigned, 
i.e., $C'$ contains all the variables in $C$ assigned as decisions.\footnote{$C'$ can also contain literals assigned by pure/unit literal detection, 
but as they are left to the maximal decision variable in the prefix, 
we treat them like decision variables. }
The literals in $C''$ are assigned due to pure and unit literal detection and may occur anywhere 
in $\prefix$ starting from $X_{i+1}$.
Further we define $\dec{C} = C'$ and $\der{C} = C''$. 
We review 
\emph{generalized model generation}~\cite{DBLP:conf/lpar/LonsingBBES15} as an
axiom to derive \nolinebreak cubes.

\begin{definition}[Generalized Model Generation~\cite{DBLP:conf/lpar/LonsingBBES15}]\label{def_generalized_model_gen}
Given a PCNF $\qclauset$ and a QCDCL assignment $A$
according to Definition~\ref{def:assign}.
If $\qclauset[A]$ is satisfiable, then
the QCDCL cube $C = (\bigwedge_{l
\in A} l)$ is obtained by
\emph{generalized model generation}.
\end{definition}

\begin{theorem}[\cite{DBLP:conf/lpar/LonsingBBES15}] \label{thm_generalized_model_gen}
Given PCNF $\qclauset = \prefix.\clauset$
and
a QCDCL cube $C$ obtained from $\qclauset$ by generalized model
generation. Then it holds that
$\prefix.\clauset \satequiv \prefix.(\clauset \vee C)$.
\end{theorem}

\begin{corollary}[\cite{DBLP:conf/lpar/LonsingBBES15}]
\label{cor:cube}
By Theorem~\ref{thm_generalized_model_gen}, a cube $C$ obtained from PCNF $\prefix. \clauset$ by generalized
model generation can be used as a learned cube in QCDCL.
\end{corollary}

Dual to generalized model generation, we define \emph{generalized conflict
generation} to derive clauses which can be added to a PCNF in a
satisfiability-preserving \nolinebreak way.

\begin{definition}[Generalized Conflict Generation]\label{def_generalized_conflict_gen}
Given a PCNF $\qclauset$ and a QCDCL assignment $A$  
according to Definition~\ref{def:assign}.
If $\qclauset[A]$ is unsatisfiable, then  
the QCDCL clause $C = (\bigvee_{l
\in A} \bar l)$ is obtained by 
\emph{generalized conflict generation}.
\end{definition}

\begin{theorem} \label{thm_generalized_conf_gen}
Given PCNF $\qclauset = \prefix.\clauset$  
and 
a QCDCL clause $C$ obtained from $\qclauset$ by generalized conflict
generation using QCDCL assignment $A$. Then it holds that 
$\prefix.\clauset \satequiv \prefix.(\clauset \wedge C)$.
\end{theorem}
\begin{proof}[Sketch]
We argue that if $\prefix.\clauset$ is satisfiable, 
so is $\prefix.(\clauset \wedge C)$. The case for unsatisfiability  is trivial. 
Let $C = C' \cup C''$ 
with $C' = \dec{C}$ and $C'' = \der{C}$. Further, let 
$A = A' \cup A''$ such that $\lit{A'} = \lit{C'}$ and 
$\lit{A''} = \lit{C''}$.
Now assume that $\prefix.\clauset$ is satisfiable, 
but $\prefix.(\clauset \wedge C)$ is not. In order to falsify $C$, 
its subclause $C'$ has to be falsified, i.e., the first $|C'|$ variables 
of $\prefix$ have to be set according to $A'$. Then, due to pure and 
unit, also $C''$ is falsified, and therefore, each assignment falsifying 
$C$ has to contain $A$. But $\prefix.\clauset[A]$ is unsatisfiable. 
Since $\prefix.\clauset$ is satisfiable, there have to be other
decisions than the decisions of $A$ to show its satisfiability, but these also satisfy 
$\prefix.(\clauset \wedge C)$.
\qed
\end{proof}

\begin{corollary}
\label{cor:clause}
By Theorem~\ref{thm_generalized_conf_gen}, a clause 
$C$ obtained from PCNF $\prefix.\clauset$ by generalized
conflict generation can be used as a learned clause in QCDCL.
\end{corollary}

Based on Corollaries~\ref{cor:cube} and~\ref{cor:clause}, we formulate
axioms to derive learned clauses (cubes) from QCDCL assignments $A$ under
which the PCNF $\qclauset$ is (un)satisfiable.

\begin{definition}[Generalized Axioms] \label{def_gen_axioms}
Let $\qclauset = \prefix.\clauset$ be a PCNF. The \emph{generalized clause} and \emph{cube
axioms} are as follows. 
\begin{align}\tag{$\operatorname{\emph{gen-cl-init}}$}
\AxiomC{\phantom{A}}
\UnaryInfC{C}
\label{rule_gen_cl_init}
\DisplayProof
\quad
\begin{minipage}{0.7\textwidth}
$A$ is a QCDCL assignment,\\ 
$\qclauset[A]$ is unsatisfiable, \\
and $C = (\bigvee_{l \in A} \bar l)$ is a QCDCL clause
\end{minipage} 
\end{align}

\vspace{-0.25cm}

\begin{align}\tag{$\operatorname{\emph{gen-cu-init}}$}
\AxiomC{\phantom{A}}
\UnaryInfC{C}
\label{rule_gen_cu_init}
\DisplayProof
\quad
\begin{minipage}{0.7\textwidth}
$A$ is a QCDCL assignment,\\ 
$\qclauset[A]$ is satisfiable, \\
and $C = (\bigwedge_{l \in A} l)$ is a QCDCL cube
\end{minipage} 
\end{align}

\end{definition}

The generalized axioms~\ref{rule_gen_cl_init} and~\ref{rule_gen_cu_init} are
added to \qrescalc in addition to the traditional axioms~\ref{rule_cl_init}
and~\ref{rule_cu_init} from Definition~\ref{def_qres_calculus}.

\begin{example}
\label{ex_better_cube_proof}
Consider the PCNF from Example~\ref{ex_bad_cube_proof}. Let $A := \{\bar z,
\bar z'\}$ be a QCDCL 

{\setlength{\parindent}{0pt}
\begin{minipage}{0.80\textwidth}
assignment where $z$ and $z'$ are assigned as
decisions. \nolinebreak The \nolinebreak PCNF $\qclauset[A] = \forall u \exists y. (u \vee \bar y)
\wedge (\bar u \vee y)$ is satisfiable. We apply \nolinebreak \mbox{axiom} \ref{rule_gen_cu_init}
to derive the cube $C := (\bar z \wedge \bar z')$ and finally the empty cube
$\ER(C) = \emptyset$ (proof shown on the right).
\end{minipage}
\hfill
\begin{minipage}{0.1\textwidth}
\flushright
\AxiomC{$\phantom{B}$}
\UnaryInfC{$\bar z \wedge \bar z'$}
\UnaryInfC{$\emptyset$}
\DisplayProof

\medskip

\hfill $\Diamond$
\end{minipage}
}
\end{example}

In contrast to axioms~\ref{rule_cl_init} and~\ref{rule_cu_init} (the latter
corresponds to \emph{model
generation}~\cite{DBLP:journals/jair/GiunchigliaNT06}), the generalized axioms
allow to derive clauses that are not part of the given PCNF
$\qclauset$ and cubes that are not implicants of the matrix of $\qclauset$.

Given the empty assignment $A = \{\}$ and a PCNF $\qclauset$, the empty clause
or cube can be derived using $A$ by axioms~\ref{rule_gen_cl_init}
or~\ref{rule_gen_cu_init} right away if $\qclauset[A]$ is unsatisfiable or
satisfiable, respectively. However, checking the satisfiability of the PCNF
$\qclauset[A]$ as required in the side conditions of the generalized axioms is
PSPACE-complete. Therefore, \emph{in practice} it is necessary to consider
non-empty QCDCL assignments $A$ and apply either complete approaches in a
bounded way, like the successful expansion-based
approaches~~\cite{DBLP:conf/fmcad/AyariB02,DBLP:conf/sat/Biere04a,DBLP:journals/jair/HeuleJLSB15,Janota20161}, 
 or incomplete polynomial-time procedures, e.g., as used in
preprocessing~\cite{DBLP:journals/jair/HeuleJLSB15}, to check the
satisfiability of $\qclauset[A]$. \emph{Sign
abstraction}~\cite{DBLP:conf/tableaux/Letz02} can be regarded as a first
approach towards more powerful cube learning as formalized by axiom~\ref{rule_gen_cu_init}.

Axioms~\ref{rule_gen_cl_init} and~\ref{rule_gen_cu_init} 
 provide a formal framework for combining Q-resolution in \qrescalc with \emph{any} QBF decision
procedure $\mathcal{D}$ by using $\mathcal{D}$ to check
$\qclauset[A]$. This framework also applies to related combinations of search-based QBF solving
with variable elimination~\cite{DBLP:conf/mbmv/ReimerPSB12}. 
Regarding proof complexity, decision procedures
like expansion and Q-resolution
are incomparable as the lengths of proofs they are
able to produce for certain 
PCNFs differ by an exponential factor~\cite{DBLP:conf/sat/BalabanovWJ14,beyersdorff_et_al:LIPIcs:2015:4905,DBLP:journals/tcs/JanotaM15}. 
Due to this property, the combination of incomparable procedures in \qrescalc
via the generalized axioms allows to
benefit from their individual strengths.  
For example, the use of
expansion to check the satisfiability of $\qclauset[A]$ in
axioms~\ref{rule_gen_cl_init} and~\ref{rule_gen_cu_init} results in a variant
of \qrescalc which is exponentially stronger than traditional \qrescalc. For
satisfiable PCNFs, QBCE, originally a preprocessing technique to eliminate
redundant clauses in a PCNF, was shown to be effective to solve $\qclauset[A]$
for applications of axiom~\ref{rule_gen_cu_init}~\cite{DBLP:conf/lpar/LonsingBBES15}, resulting
in an exponentially stronger cube proof system.

If a decision procedure $\mathcal{D}$ is applied as a black box to check
$\qclauset[A]$, then \qrescalc extended by~\ref{rule_gen_cl_init}
and~\ref{rule_gen_cu_init} is not a proof system as defined by Cook and
Reckhow~\cite{DBLP:journals/jsyml/CookR79} because the final proof $P$ of
$\qclauset$ cannot be checked in polynomial time. However, $\mathcal{D}$ can
be augmented to return a proof $P'$ of $\qclauset[A]$ for every application
of~\ref{rule_gen_cl_init} and~\ref{rule_gen_cu_init}. Such proof $P'$ may be
formulated, e.g., in the QRAT proof system~\cite{DBLP:conf/cade/HeuleSB14}.
Finally, the proof $P$ of $\qclauset$ contains subproofs $P'$, all of which can
be checked in polynomial time, like $P$ itself (the size of $P$ may blow up exponentially in the worst case 
depending on the decision procedures that are used to produce the subproofs $P'$).

The QCDCL framework (Fig.~\ref{fig:qcdcl:gen:axioms}) readily supports
applications of the generalized axioms~\ref{rule_gen_cl_init}
and~\ref{rule_gen_cu_init}. A clause (resp.~cube) $C$ derived by these
axioms is first reduced by universal (resp.~existential) reduction to
obtain a reduced clause (cube) $C' \subseteq C$. Then $C'$ is used to derive
an asserting learned clause (cube) \emph{in the same way} as in clause
learning by traditional \qrescalc (Definition~\ref{def_qres_calculus}).

\section{An Abstraction-Based Clause Axiom}
\label{sec:abs:axiom}

Axioms~\ref{rule_gen_cl_init} and~\ref{rule_gen_cu_init} by
Definition~\ref{def_gen_axioms} are based on QCDCL assignments, where decision
variables have to be assigned in prefix ordering. To overcome the order
restriction, we introduce a clause axiom which allows to derive clauses based
on an abstraction of a PCNF and \emph{arbitrary} assignments.

\begin{definition}[Existential Abstraction]
\label{def_eabs}
Let $\qclauset = \prefix.\clauset$ be a PCNF with prefix $\prefix :=
    Q_1X_1 Q_2X_2 \ldots Q_nX_n$ and matrix $\clauset$. The \emph{existential
      abstraction} $\eabs{\qclauset} := \prefix'\!.\clauset$ of $\qclauset$ has 
prefix $\prefix' := \exists (X_1 \cup X_2 \cup  \ldots \cup X_n)$.
\end{definition}
\begin{lemma}\label{lem:abs:prop:model:preserving}
Let $\qclauset = \prefix.\clauset$ be a PCNF, $\eabs{\qclauset}$ its
existential abstraction, and $A$ a partial assignment of the variables in
$\eabs{\qclauset}$.  If 
$\eabs{\qclauset}[A]$ is unsatisfiable then ${\clauset} \equiv
{\clauset} \wedge (\bigvee_{l \in A} \bar l)$. 
\end{lemma}
\begin{proof}
Obviously, every model $M$ of ${\clauset} \wedge (\bigvee_{l \in A}
\bar l)$ is also a model of ${\clauset}$. To show the other direction,
let $M$ be a model of $\clauset$, but (${\clauset} \wedge (\bigvee_{l \in A}
\bar l))[M] = \false$. Then $A \subseteq M$. Since 
$\eabs{\qclauset}[A]$ is unsatisfiable, also $\clauset[A]$ is 
unsatisfiable. Then $M$ cannot be a model of $\clauset$.
\qed
\end{proof}

\begin{theorem}[cf.~\cite{DBLP:conf/cp/SamulowitzB05,DBLP:conf/cp/SamulowitzDB06}] 
\label{thm_abs_conf_gen_correctness}
For a PCNF $\qclauset = \prefix.\clauset$, $\eabs{\qclauset}$ its
existential abstraction, and a partial assignment $A$ of the variables in
$\eabs{\qclauset}$ such that $\eabs{\qclauset}[A]$ is unsatisfiable, it holds
that $\prefix.\clauset \satequiv \prefix.(\clauset \wedge (\bigvee_{l \in A} \bar l))$. 
\end{theorem}
\begin{proof}
By Lemma~\ref{lem:abs:prop:model:preserving}, ${\clauset}$ and
${\clauset} \wedge (\bigvee_{l \in A} \bar l)$ have the same sets of
propositional models. As argued in the context of SAT-based QBF
solving~\cite{DBLP:conf/cp/SamulowitzB05} and QBF
preprocessing~\cite{DBLP:conf/cp/SamulowitzDB06}, model-preserving
manipulations of the matrix of a PCNF result in a satisfiability-equivalent
PCNF.\footnote{In fact, a stronger result is proved
in~\cite{DBLP:conf/cp/SamulowitzDB06}: model-preserving manipulations of the
matrix of a PCNF result in a PCNF having the same set of \emph{tree-like QBF
models}.}
\qed
\end{proof}

\begin{definition}[Abstraction-Based Conflict Generation]\label{def_abs_conf_gen}
Given a PCNF $\qclauset$, its existential abstraction $\eabs{\qclauset}$ and
an assignment $A$ (not necessarily being a QCDCL assignment).  
If $\eabs{\qclauset}[A]$ is unsatisfiable, then  
the clause $C = (\bigvee_{l
\in A} \bar l)$ is obtained by 
\emph{abstraction-based conflict generation}.
\end{definition}

We formulate a new axiom to derive clauses by abstraction-based conflict
generation, which can be used as ordinary learned clauses in QCDCL
(Theorem~\ref{thm_abs_conf_gen_correctness}).

\begin{definition}[Abstraction-Based Clause Axiom] \label{def_abs_gen}
For a PCNF $\qclauset = \prefix.\clauset$ and $\eabs{\qclauset}$ by
Definition~\ref{def_eabs}, the \emph{abstraction-based clause axiom} is as
follows:
\begin{align}\tag{$\operatorname{\emph{abs-cl-init}}$}
\AxiomC{\phantom{A}}
\UnaryInfC{C}
\label{rule_abs_cl_init}
\DisplayProof
\quad
\begin{minipage}{0.7\textwidth}
$A$ is an assignment,\\ 
$\eabs{\qclauset}[A]$ is unsatisfiable, \\
and $C = (\bigvee_{l \in A} \bar l)$ is a clause
\end{minipage} 
\end{align}
\end{definition}

Axiom~\ref{rule_abs_cl_init} can be added to \qrescalc in addition to all the
other axioms. In the side condition of axiom~\ref{rule_abs_cl_init}, the
propositional CNF $\eabs{\qclauset}[A]$ has to be solved, which naturally can
be carried out by integrating a SAT solver in QCDCL. SAT solving has been
applied in the context of QCDCL to derive
learned clauses~\cite{DBLP:conf/cp/SamulowitzB05} and to 
overcome the ordering of the prefix of a PCNF. Further, many QBF
solvers rely on  SAT
solving~\cite{DBLP:conf/ijcai/JanotaM15,Janota20161,DBLP:conf/fmcad/RabeT15,DBLP:conf/sat/TuHJ15}.
Integrating axiom~\ref{rule_abs_cl_init} in \qrescalc by
Definition~\ref{def_qres_calculus} results in a variant of \qrescalc which is
exponentially stronger than traditional \qrescalc, as illustrated by the
following example.

\begin{example}
\label{ex_hkb}
Consider the following family $(\hkb)_{t \geq 1}$ of PCNFs defined by Kleine Büning et
al.~\cite{DBLP:journals/iandc/BuningKF95}.  
A formula $\hkb$ in $(\hkb)_{t \geq 1}$  
 has the quantifier prefix 
$$
\exists d_0 d_1 \eVar_1 \forall \uVar_1 \exists d_2 \eVar_2 \forall \uVar_2 \exists
d_3 \eVar_3\ldots \forall \uVar_{t-1} \exists d_t \eVar_t \forall \uVar_t \exists
f_{1}\ldots f_{t}
$$
and a matrix consisting of the following clauses:
\[
\begin{array}{lclclclcl}
\clause_0      & :=  & \overline{d}_0 &  & 
\clause_1      & := & d_0 \lor \overline{d}_1 \lor \eVarNeg_1 \\
\clause_{2j}   & := & d_j \lor \uVarNeg_j \lor \overline{d}_{j+1} 
\lor \eVarNeg_{j+1} & & 
\clause_{2j+1} & := & e_j \lor \uVar_j \lor \overline{d}_{j+1} 
\lor \eVarNeg_{j+1}
&  &
	\text{for } 1 \leq j < t\\
\clause_{2t}   & := & d_t \lor \uVarNeg_{t} \lor \overline{f}_{1}\lor \ldots  
\lor \overline{f}_{t}
& &  
\clause_{2t+1} & := & \eVar_t \lor \uVar_{t} \lor \overline{f}_{1}\lor \ldots  \lor 
\overline{f}_{t}\\
B_{2j-1} & := & \uVar_j \lor f_j & & 
B_{2j}   & := &  \uVarNeg_j \lor f_{j} & & \text{for } 1 \leq j \leq t
\end{array}
\]
The size of every clause resolution proof of $\hkb$ in traditional \qrescalc
(Definition~\ref{def_qres_calculus}) is exponential in
$t$~\cite{beyersdorff_et_al:LIPIcs:2015:4905,DBLP:journals/iandc/BuningKF95}. We
show that \qrescalc with axiom~\ref{rule_abs_cl_init} allows to generate
proofs of $\hkb$ which are polynomial in $t$.  To this end, we
apply~\ref{rule_abs_cl_init} to derive unit clauses $(f_j)$ for all
existential variables $f_j$ in $\hkb$ using assignments $A := \{ \bar f_{j}
\}$, respectively.  Since $\eabs{\hkb}[A]$ contains complementary unit
clauses $(\uVar_j)$ and $(\uVarNeg_j)$ resulting from the clauses $B_{2j-1}$
and $B_{2j}$ in $\hkb$, the unsatisfiability of $\eabs{\hkb}[A]$ can be
determined in polynomial time without invoking a SAT~solver.  The derived unit clauses $(f_j)$ are resolved
with clauses $\clause_{2t}$ and $\clause_{2t+1}$ to produce further unit
clauses $(d_t)$ and $(e_t)$ after universal reduction. This process continues
with $\clause_{2j}$ and $\clause_{2j+1}$ until the empty clause is
derived using $C_0$ and $C_1$. \hfill $\Diamond$
\end{example}

Abstraction-based failed literal detection~\cite{DBLP:conf/sat/LonsingB11},
where certain universal quantifiers of a PCNF are treated as existential ones,
implicitly relies on \emph{QU-resolution}. QU-resolution allows universal
variables as pivots in rule~\ref{rule_res} and can generate the same proofs of
$(\hkb)_{t \geq 1}$ as
in Example~\ref{ex_hkb}~\cite{DBLP:conf/cp/Gelder12}. Applying
axiom~\ref{rule_abs_cl_init} for clause learning in QCDCL harnesses the power
of SAT solving. Furthermore, the combination of \qrescalc
(Definition~\ref{def_qres_calculus}) and~\ref{rule_abs_cl_init}
polynomially simulates\footnote{We refer to the appendix of this paper with additional
  results.} QU-resolution, which has not been
applied systematically to learn clauses in QCDCL. Like with the axioms~\ref{rule_gen_cl_init}
and~\ref{rule_gen_cu_init}, clauses derived by axiom~\ref{rule_abs_cl_init}
can readily be used to derive asserting learned clauses in QCDCL.


\section{Case Study and Experiments}
\label{sec:exp}

\depqbf\footnote{DepQBF is free software: \url{http://lonsing.github.io/depqbf/}} is a QCDCL-based QBF solver
implementing the Q-resolution calculus as in
Definition~\ref{def_qres_calculus}. Since version 5.0, \depqbf
additionally applies the generalized cube axiom~\ref{rule_gen_cu_init}
based on dynamic blocked clause elimination
(QBCE)~\cite{DBLP:conf/lpar/LonsingBBES15}. The case where QBCE
reduces the PCNF $\qclauset[A]$ under the current assignment $A$ to the empty formula constitutes a
successful application of axiom~\ref{rule_gen_cu_init}.
\depqbf comes with a sophisticated analysis of variable dependencies in a
PCNF~\cite{DBLP:journals/jar/SamerS09} to relax their linear prefix ordering.
However, we disabled dependency analysis 
to focus the evaluation on axiom applications.
In the following, 
we evaluate the impact of (combinations of) the generalized 
axioms~\ref{rule_gen_cl_init} and~\ref{rule_gen_cu_init} and the
abstraction-based clause axiom~\ref{rule_abs_cl_init} in practice. 


\subsection{Axiom Applications in Practice}

In \depqbf, we attempt to apply the generalized axioms after QBCP has
saturated in QCDCL, i.e., before assigning a variable as decision. 
We integrated the preprocessor Bloqqer~\cite{DBLP:conf/cade/BiereLS11}
to detect applications of~\ref{rule_gen_cl_init}
and~\ref{rule_gen_cu_init}. Bloqqer implements techniques such as equivalence
reasoning, variable elimination, (variants of) QBCE, and expansion of
universal variables. Since these techniques are applied in bounded fashion,
Bloqqer can be regarded as an incomplete QBF solver. If the PCNF $\qclauset[A]$ is satisfiable
(unsatisfiable) and Bloqqer solves it, then a QCDCL cube (clause) is generated 
by axiom~\ref{rule_gen_cu_init} (\ref{rule_gen_cl_init}), which is used to derive a
learned cube (clause). Otherwise, QCDCL
proceeds as usual with assigning a decision variable. 
Bloqqer is explicitly provided with the entire PCNF $\qclauset[A]$ before each
call. To limit the resulting run time overhead in practice, Bloqqer is called
in intervals of $2^n$ decisions, where $n := 11$ in
our experiments. Further, Bloqqer is never called on PCNFs with more
than 500,000 original clauses, and it is disabled at run time if the average
time spent to complete a call exceeds 0.125 seconds.

To detect applications of the abstraction-based clause
axiom~\ref{rule_abs_cl_init}, we use the SAT solver
PicoSAT~\cite{DBLP:journals/jsat/Biere08} to check the satisfiability of the
existential abstraction $\eabs{\qclauset}[A]$ of the PCNF $\qclauset =
\prefix. \clauset$ under
the current QCDCL assignment $A$. The matrix $\clauset$ is imported to PicoSAT
once before the entire solving process starts. For each check of $\eabs{\qclauset}[A]$, the
QCDCL assignment $A$ is passed to PicoSAT via assumptions, and PicoSAT is
called incrementally. If $\eabs{\qclauset}[A]$ is unsatisfiable, then we try
to minimize the size of $A$ by extracting the set $A' \subseteq A$ of
\emph{failed assumptions}. Failed assumptions are those assumptions that were
relevant for the SAT solver to determine the unsatisfiability of
$\eabs{\qclauset}[A]$. Note that in general $A'$ is not a QCDCL assignment. It holds that $\eabs{\qclauset}[A']$ is unsatisfiable
and hence we derive the clause $C = (\bigvee_{l \in A'} \bar l)$ by
axiom~\ref{rule_abs_cl_init}.

In addition to Bloqqer and dynamic QBCE (which is part of
DepQBF~5.0~\cite{DBLP:conf/lpar/LonsingBBES15}) used to detect applications of
the generalized cube axiom~\ref{rule_gen_cu_init}, we implemented a
\emph{trivial truth}~\cite{DBLP:conf/aaai/CadoliGS98} test based on the
following abstraction.
\begin{definition}[Universal Literal Abstraction, cf.~\emph{Trivial Truth}~\cite{DBLP:conf/aaai/CadoliGS98}]
\label{def_univ_lit_abs}
Let $\qclauset = \prefix.\clauset$ be a PCNF. The \emph{universal
  literal abstraction} $\aabs{\qclauset} := \prefix'\!.\clauset'$ of
$\qclauset$ is obtained by removing all universal literals from all
the clauses in $\clauset$ and by removing all universal variables and
universal quantifiers from $\prefix$.
\end{definition}
\begin{lemma}[\cite{DBLP:conf/aaai/CadoliGS98}]
\label{lem_trivial_truth_correctness}
For a PCNF $\qclauset = \prefix.\clauset$, $\aabs{\qclauset}$, and a QCDCL assignment $A$ of variables in
$\aabs{\qclauset}$: if $\aabs{\qclauset}[A]$ is satisfiable, then $\qclauset[A]$ is satisfiable. 
\end{lemma}
By Lemma~\ref{lem_trivial_truth_correctness}, we can check the side condition
of axiom~\ref{rule_gen_cu_init} whether $\qclauset[A]$ is satisfiable under a
QCDCL assignment $A$ by checking whether $\aabs{\qclauset}[A]$ is satisfiable.
To this end, we use a second instance of PicoSAT. Note that while
Definition~\ref{def_univ_lit_abs} corresponds to trivial truth, the
existential abstraction (Definition~\ref{def_eabs}) corresponds to
\emph{trivial falsity}~\cite{DBLP:conf/aaai/CadoliGS98}. Hence by axiom
applications, we apply trivial truth and falsity, which originate from purely
search-based QBF solving without learning, to derive clauses and cubes in
QCDCL.

Like Bloqqer, we call the two instances of PicoSAT to detect applications
of~\ref{rule_abs_cl_init} and~\ref{rule_gen_cu_init} in QCDCL before assigning
a decision variable. PicoSAT is called in intervals of $2^m$ decisions,
where $m := 10$. PicoSAT is never called on PCNFs with more than 500,000
original clauses, and it is disabled at run time if the average time spent to complete
a call exceeds five seconds. 

\subsection{Experimental Results}

The integration of Bloqqer and SAT solving to detect axiom applications
results in several variants of \depqbf. We use the letter code
``\depqbfOnlyDynQBCE-\{nQ$\mid$B$\mid$A$\mid$T\}'' to label the variants,
where ``DQ'' represents \depqbf 5.0 with dynamic QBCE used for
axiom~\ref{rule_gen_cu_init}~\cite{DBLP:conf/lpar/LonsingBBES15}. Variant
``nQ'' indicates that dynamic QBCE is disabled. Letters, ``B'', ``A'', and
``T'' represent the additional application of Bloqqer for
axioms~\ref{rule_gen_cl_init} and~\ref{rule_gen_cu_init}, SAT solving to check
the existential abstraction for axiom~\ref{rule_abs_cl_init}, and SAT solving
to carry out the trivial truth test for~\ref{rule_gen_cu_init}, respectively.

\begin{table}[t]
{
\begin{minipage}[b]{0.48\textwidth}
\caption{Preprocessing track. Solved
  instances (\emph{\#T}), solved unsatisfiable (\emph{\#U}) and satisfiable
  ones (\emph{\#S}), and total wall clock time in seconds including time outs.}
\centering
\begin{tabular}{l@{\enskip\enskip}r@{\enskip}r@{\enskip}r@{\enskip}r}
\hline
\emph{Solver} & \multicolumn{1}{c}{\emph{\#T}} & \multicolumn{1}{c}{\emph{\#U}} & \multicolumn{1}{c}{\emph{\#S}} & \multicolumn{1}{c}{\emph{Time}} \\
\hline
\rareqs & 107 & 44 & 63 & 255K \\
\depqbfNoQBCETFTT  & 105 & 46 & 59 & 266K \\
\qesto & 104 & 46 & 58 & 267K \\
\depqbfNoQBCE  & 101 & 44 & 57 & 271K \\
\depqbfDynTFTT & 99 & 45 & 54 & 273K \\
\depqbfDynBloqqerTFTT & 98 & 43 & 55 & 276K \\
\depqbfOnlyDynQBCE & 95 & 43 & 52 & 278K \\
\depqbfOnlyDynTF & 95 & 44 & 51 & 280K \\
\depqbfOnlyDynTT & 94 & 41 & 53 & 278K \\
\depqbfOnlyDynBloqqer & 94 & 42 & 52 & 284K \\
\qellc & 87 & 34 & 53 & 290K \\
\caqe & 74 & 24 & 50 & 319K \\
\ghostq & 61 & 18 & 43 & 338K \\
\hline
\end{tabular}
\label{tab:prepro:track}
\end{minipage}
\hfill
\begin{minipage}[b]{0.48\textwidth}
\caption{QBFLIB track. Same column headers as Table~\ref{tab:prepro:track}.}
\vspace*{0.83cm}
\centering
\begin{tabular}{l@{\enskip\enskip}r@{\enskip}r@{\enskip}r@{\enskip}r}
\hline
\emph{Solver} & \multicolumn{1}{c}{\emph{\#T}} & \multicolumn{1}{c}{\emph{\#U}} & \multicolumn{1}{c}{\emph{\#S}} & \multicolumn{1}{c}{\emph{Time}} \\
\hline
\ghostq & 139 & 62 & 77 & 265K \\
\depqbfDynTFTT & 110 & 58 & 52 & 314K \\
\depqbfDynBloqqerTFTT & 109 & 56 & 53 & 314K \\
\depqbfOnlyDynTT & 108 & 56 & 52 & 318K \\
\qellc & 106 & 48 & 58 & 320K \\
\depqbfOnlyDynTF & 106 & 58 & 48 & 321K \\
\depqbfOnlyDynQBCE & 105 & 57 & 48 & 326K \\
\depqbfOnlyDynBloqqer & 104 & 56 & 48 & 326K \\
\depqbfNoQBCETFTT & 88 & 49 & 39 & 352K \\
\depqbfNoQBCE  & 82 & 44 & 38 & 362K \\
\rareqs & 80 & 47 & 33 & 361K \\
\qesto & 73 & 46 & 27 & 378K \\
\caqe & 53 & 32 & 21 & 406K \\
\hline
\end{tabular}
\label{tab:qbflib:track}
\end{minipage}
}
\end{table}

For the empirical evaluation, we used the original benchmark sets from the QBF
Gallery
2014~\cite{gallery14}\footnote{\url{http://qbf.satisfiability.org/gallery/}}
preprocessing track (243 instances), QBFLIB track (276 instances), and
applications track (735 instances).  We compare the variants of \depqbf to
\rareqs~\cite{Janota20161} and \ghostq~\cite{DBLP:conf/sat/KlieberSGC10},
which showed top performance in the QBF Gallery 2014, and to the recent
solvers \caqe~\cite{DBLP:conf/fmcad/RabeT15}\footnote{The
authors~\cite{DBLP:conf/fmcad/RabeT15} provided us with an updated version
which we used in our tests.}, \qesto~\cite{DBLP:conf/ijcai/JanotaM15}, and
\qell~\cite{DBLP:conf/sat/TuHJ15}. We tested \qell with (\qellc) and without
(\qell-nc) exploiting circuit information and show only the results of the
better variant of the two in terms of solved instances.  All experiments
reported in the following were run on an AMD Opteron 6238 at 2.6 GHz under
64-bit Ubuntu Linux 12.04 with time and memory limits of 1800 seconds and 7
GB, respectively.

Tables~\ref{tab:prepro:track} to~\ref{tab:applications:no:bloqqer} illustrate
solver performance by solved instances and total wall clock time. For
\depqbf, the variant where only dynamic QBCE is applied (\depqbfOnlyDynQBCE)
is the baseline of the comparison. In the QBFLIB
(Table~\ref{tab:qbflib:track}) and applications track
(Table~\ref{tab:applications:no:bloqqer}), \depqbf with Bloqqer and SAT
solving for axioms~\ref{rule_gen_cl_init}, \ref{rule_gen_cu_init},
and~\ref{rule_abs_cl_init} solves substantially more instances than
\depqbfOnlyDynQBCE.

Disabling dynamic QBCE used for axiom~\ref{rule_gen_cu_init} (variants with
``nQ'' in the tables) results in a considerable performance decrease, except
in the preprocessing track (Table~\ref{tab:prepro:track}). There, dynamic QBCE
is harmful to the performance. We attribute this phenomenon to 
massive preprocessing, after which QBCE does not pay off. However, SAT
solving for axioms~\ref{rule_gen_cl_init} and~\ref{rule_gen_cu_init} is
crucial as the variant \depqbfNoQBCETFTT outperforms \depqbfNoQBCE without SAT
solving.

In general, combinations of dynamic QBCE, Bloqqer, and SAT solving (for
solving the existential abstraction and for testing trivial truth) outperform
variants where only one of these techniques is applied. Examples
are~\depqbfDynTFTT, \depqbfOnlyDynTF and \depqbfOnlyDynTT in
Table~\ref{tab:qbflib:track} and \depqbfDynBloqqerTFTT, \depqbfOnlyDynBloqqer,
\depqbfOnlyDynTF, and \depqbfOnlyDynTT in
Table~\ref{tab:applications:no:bloqqer}. The results in the applications track
are most pronounced, where six out of eight variants of DepQBF
outperform the other solvers (Fig.~\ref{plot:applications:no:bloqqer} shows
a related cactus plot of the run times). In the following we focus on the
applications track.


\begin{table}[t]
\begin{minipage}{0.42\textwidth}
\caption{Applications track. Same column headers as Table~\ref{tab:prepro:track}.}
\centering
    \begin{tabular}{l@{\enskip\enskip}c@{\enskip}c@{\enskip}c@{\enskip}c}
      \hline
\emph{Solver} & \emph{\#T} & \emph{\#U} & \emph{\#S} & \emph{Time} \\
      \hline
      \depqbfDynBloqqerTFTT & 466 & 236 & 230 & 553K \\
      \depqbfDynTFTT & 461 & 234 & 227 & 555K \\
      \depqbfOnlyDynTF & 459 & 237 & 222 & 561K \\
      \depqbfOnlyDynBloqqer & 449 & 222 & 227 & 563K \\
      \depqbfOnlyDynTT & 441 & 220 & 221 & 571K \\
      \depqbfOnlyDynQBCE & 441 & 224 & 217 & 575K \\
      \qellnc & 434 & 302 & 132 & 563K \\
      \rareqs & 414 & 272 & 142 & 611K \\
      \caqe & 370 & 192 & 178 & 708K \\
      \ghostq & 347 & 166 & 181 & 752K \\
      \qesto & 331 & 188 & 143 & 767K \\
     \depqbfNoQBCEDynBloqqerTFTT  & 293 & 140 & 153 & 848K \\
     \depqbfNoQBCE  & 279 & 127 & 152 & 880K \\
      \hline
    \end{tabular}
\label{tab:applications:no:bloqqer}
\end{minipage}
\hfill
\begin{minipage}{0.56\textwidth}
\centering
    \includegraphics[scale=0.63]{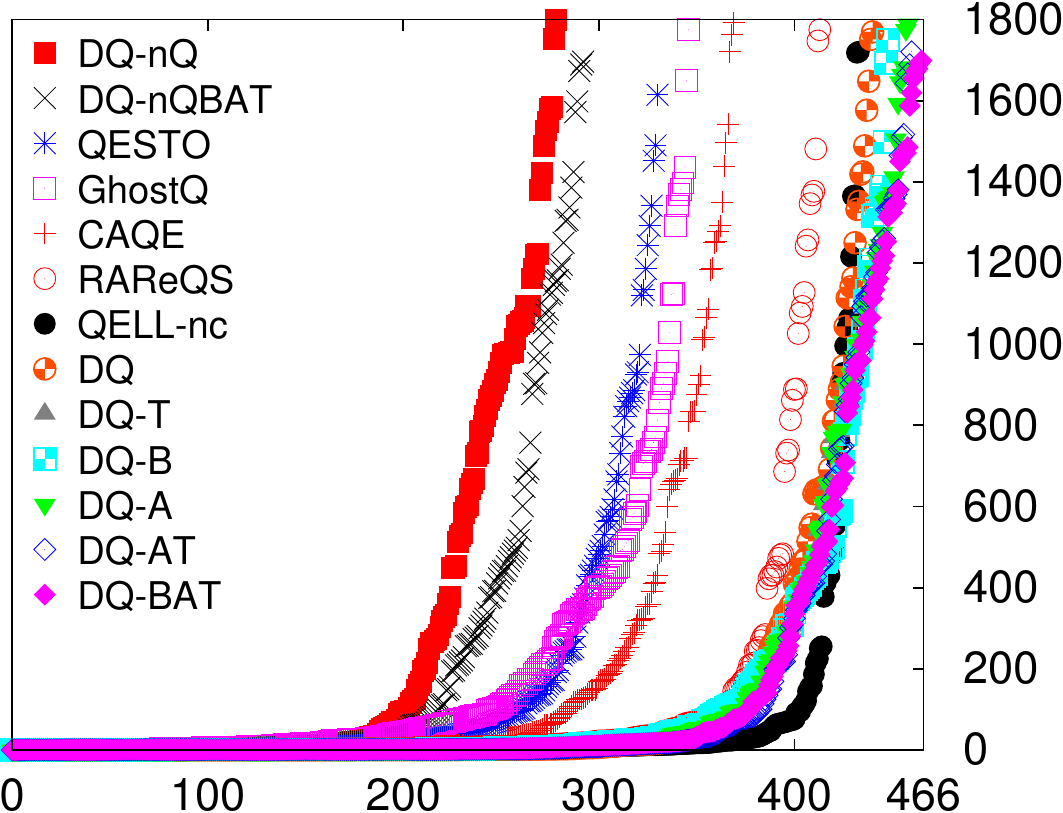} 
    \captionof{figure}{Sorted run times (y-axis) of instances (x-axis) related to Table~\ref{tab:applications:no:bloqqer}.}
    \label{plot:applications:no:bloqqer}
\end{minipage}
\end{table}

\begin{table}[t]
\begin{minipage}[c]{0.47\textwidth}
\caption{Related to variant \depqbfDynBloqqerTFTT in
  Table~\ref{tab:applications:no:bloqqer}: statistics on applications of
  Bloqqer (B), SAT solving for~\ref{rule_abs_cl_init} (A), and SAT solving to
  test trivial truth for~\ref{rule_gen_cu_init} (T) with respect to total solved
  instances (\emph{\#T}) and solved satisfiable (\emph{\#S}) and unsatisfiable
  ones (\emph{\#U}).}
\centering
    \begin{tabular}{lrrr}
\hline
 & \multicolumn{1}{c}{\emph{\#T}} & \multicolumn{1}{c}{\emph{\#S}} & \multicolumn{1}{c}{\emph{\#U}} \\
\hline
B tried: & 18559 & 12052 & 6507 \\
B success: & 18150 & 11946 & 6204 \\
B sat: & 10917 & 10405 & 512 \\
B unsat: & 7233 & 1541 & 5692 \\
\hline
T tried: & 241,180 & 88,623 & 152,557 \\
T success: & 20,494 & 19,276 &  1,218 \\
\hline
A tried: & 301,652 & 122,929 & 178,723 \\
A success: & 67,129 & 34,306 & 32,823 \\
\hline
    \end{tabular}
\label{tab:dyn:bloqqer:tt:tf:stats}
\end{minipage}
\hfill
\begin{minipage}[c]{0.47\textwidth}
\centering
\includegraphics[scale=0.48]{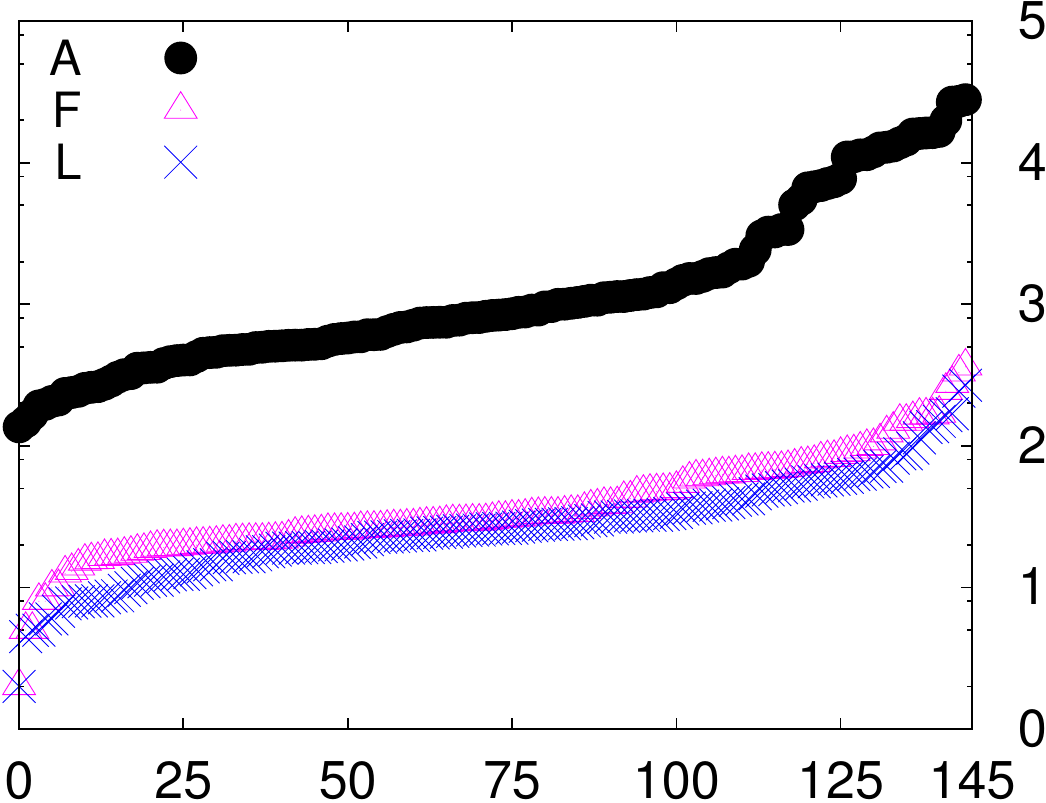}
\captionof{figure}{Average SAT solver assumptions per successful
application of~\ref{rule_abs_cl_init} (``A'') on 145 selected instances solved 
by \depqbfDynBloqqerTFTT (Table~\ref{tab:applications:no:bloqqer}), 
failed assumptions (``F''), and literals in the clauses
learned by~\ref{rule_abs_cl_init} (``L''), $\mathit{log}_{10}$
scale on y-axis.}
\label{fig:dyn:bloqqer:tt:tf:assumptions}
\end{minipage}
\end{table}

Consider the best performing variant \depqbfDynBloqqerTFTT in
Table~\ref{tab:applications:no:bloqqer}. Table~\ref{tab:dyn:bloqqer:tt:tf:stats}
shows statistics on the number of attempted and successful applications of
axioms~\ref{rule_gen_cl_init}, \ref{rule_gen_cu_init}
and~\ref{rule_abs_cl_init} by Bloqqer and SAT solving. On the 466 instances
solved by \depqbfDynBloqqerTFTT, Bloqqer was called on $\qclauset[A]$ at least once
on 185 instances and successfully solved $\qclauset[A]$ at least once on 184
instances, thus allowing applications of axiom~\ref{rule_gen_cl_init}
or~\ref{rule_gen_cu_init}. Bloqqer was disabled at run time on 143 instances
due to the predefined limits. SAT solving for the trivial truth test for
\ref{rule_gen_cu_init} (respectively, to solve the existential abstraction
for~\ref{rule_abs_cl_init}) was applied at least once on 364 (445) instances,
was successful at least once on 177 (226) instances, and was disabled at run
time on 21 (70) instances. While Bloqqer is applied less frequently than
SAT solving by a factor of two,
applications of Bloqqer have much higher success rates ($97\%$) than SAT
solving ($8\%$ and $22\%$).

In the following, we analyze applications of the abstraction-based clause
axiom in more detail.  The extraction of failed assumptions in SAT solving
for~\ref{rule_abs_cl_init} allows to reduce the size of the clauses learned by
abstraction-based conflict generation. On 145 instances solved by
\depqbfDynBloqqerTFTT (Table~\ref{tab:applications:no:bloqqer}),
axiom~\ref{rule_abs_cl_init} was applied more than once. Per instance, on
average (median) 3,336K (70.7K) assumptions were passed to the SAT
solver when solving $\eabs{\qclauset}[A]$, 28.8K (2.3K) failed
assumptions were extracted, and the clauses finally learned had 20.7K (1.5K) 
literals. The difference in the number of failed assumptions and the
size of learned clauses is due to additional, heuristic minimization of the set of failed
assumptions which we apply. Given that $\eabs{\qclauset}[A]$ is unsatisfiable, it may be
possible to remove assignments from $A$, thus resulting in a smaller
assignment $A'$, while preserving unsatisfiability of
$\eabs{\qclauset}[A']$. Additionally, universal reduction by
rule~\ref{rule_red} may remove literals from the clause learned by generalized
conflict generation.  Figure~\ref{fig:dyn:bloqqer:tt:tf:assumptions} shows
related average statistics.

The abstraction-based clause axiom~\ref{rule_abs_cl_init} is particularly
effective on instances from the domain of conformant planning. With variant
\depqbfDynBloqqerTFTT (Table~\ref{tab:applications:no:bloqqer}), 81
unsatisfiable instances from conformant planning were solved by a
\emph{single} application of axiom~\ref{rule_abs_cl_init} where the empty
clause was derived immediately.  On 13 of these 81
instances, solving $\eabs{\qclauset}$ was hard for the SAT
solver, which took more than 900 seconds. In contrast to
\depqbfDynBloqqerTFTT, \depqbfOnlyDynQBCE does not use
axiom~\ref{rule_abs_cl_init} and failed to solve 15 of the 81 instances.

Additionally, we evaluated the variants of \depqbf and the other solvers on
the benchmarks of the applications and QBFLIB tracks \emph{with} preprocessing
by \bloqqer before
solving.\footnote{We refer to the appendix of this paper with additional tables.} In the QBFLIB track,
\rareqs and \depqbfOnlyDynTT solved the largest number of instances ($134$ in
total each instead of $80$ and $108$ in Table~\ref{tab:qbflib:track}).
However, here it is important to remark that already the plain variant \depqbf
solved $132$ instances if \bloqqer is applied before solving.  With partial
preprocessing by \bloqqer (using only QBCE and universal expansion), on
the applications track \qellnc and \depqbfDynTFTT each solved $483$ instances,
i.e., $49$ and $22$ more instances than without preprocessing
(Table~\ref{tab:applications:no:bloqqer}). Note that the best variant
\depqbfDynBloqqerTFTT of \depqbf in Table~\ref{tab:applications:no:bloqqer}
solved $480$ instances. Partial preprocessing increases the number of
instances solved by the variants of \depqbf. In contrast to that, with full
preprocessing the performance of the variants of \depqbf on the applications
track considerably decreases.  If \bloqqer is applied to the full extent
(enabling all techniques), then \rareqs, \qellnc, and \qesto solve $547$,
$501$, and $463$ instances, respectively. The variant \depqbfDynTFTT of
\depqbf, however, which solved $483$ instances with partial preprocessing,
solves only $434$ instances.  The phenomenon that preprocessing is not always
beneficial was also observed in the QBF
Galleries~\cite{gallery14,Lonsing201692}.  When applied
without restrictions, \bloqqer rewrites a formula and thus destroys or blurs
structural information. For some approaches structural information is essential to
fully exploit their individual strengths.


\section{Conclusion}
\label{sec:concl}

The Q-resolution calculus \qrescalc is a proof system which underlies
clause and cube learning in QCDCL-based QBF solvers. In QCDCL, the
traditional axioms of \qrescalc either select clauses which already
appear in the input PCNF $\qclauset$ or construct cubes which are
implicants of the matrix of $\qclauset$.

To overcome the limited deductive power of the traditional axioms, we
presented two generalized axioms to derive clauses and
cubes based on checking the satisfiability of $\qclauset$ under an
assignment $A$ generated in QCDCL. We also
formulated a new axiom to derive clauses which relies on an
existential abstraction of $\qclauset$ and on SAT solving. This
abstraction-based axiom leverages QU-resolution and allows to
overcome the prefix order restriction in QCDCL to some extent. The new axioms can be
integrated in \qrescalc and used for clause and cube learning 
in the QCDCL framework. They are  
compatible with any variant of Q-resolution, like
long-distance resolution~\cite{DBLP:conf/cp/ZhangM02},
QU-resolution~\cite{DBLP:conf/cp/Gelder12}, and combinations
thereof~\cite{DBLP:conf/sat/BalabanovWJ14}. 

For axiom applications in practice, \emph{any}
complete or incomplete QBF decision procedure can be applied to check
the satisfiability of $\qclauset$ under assignment $A$. In this respect, the
generalized axioms act as an interface to combining Q-resolution with
other QBF decision procedures in \qrescalc. 
The combination of orthogonal techniques like expansion via the
generalized axioms results in variants of \qrescalc which are stronger
than traditional \qrescalc with respect to proof complexity. 
A proof $P$ produced by such variants of \qrescalc
can be checked in time which is polynomial in the size of $P$ if subproofs of all clauses and cubes
derived by the generalized axioms are provided by the QBF decision procedures.

In order to demonstrate the effectiveness of the newly introduced
axioms, we made case studies using the QCDCL solver \depqbf. We
applied the preprocessor \bloqqer and SAT solving as incomplete QBF
decision procedures in \depqbf to detect axiom
applications. Overall, our experiments showed a considerable
performance improvement of QCDCL, particularly on application
instances.

As future work, it would be interesting to integrate techniques like
expansion-based QBF solving more tightly in QCDCL than what
we achieved with \bloqqer in our case study. A tighter integration
would allow to reduce the run time overhead we observed in
practice. Further research directions include axiom applications based
on different QBF solving techniques in parallel QCDCL, and potential
relaxations of the prefix order in assignments used for
axiom applications.



\clearpage
\newpage

\begin{appendix}

\section{Simulation Result}

\newcommand{\qrescalcabs}{QRES-abs\xspace}

We show that the combination of
\qrescalc and the abstraction-based axiom~\ref{rule_abs_cl_init} polynomially simulates
QU-resolution.

\begin{definition}[Q-Resolution with Abstraction-based
Axiom] \label{def_qres_calc_with_abs_axiom} The combination of the
Q-resolution calculus \qrescalc (Definition~\ref{def_qres_calculus}) and the
abstraction-based axiom~\ref{rule_abs_cl_init} (Definition~\ref{def_abs_gen})
is denoted by \emph{\qrescalcabs.}
\end{definition}

Note that \qrescalcabs includes all the rules of \qrescalc and additionally
the abstraction-based axiom~\ref{rule_abs_cl_init}.

\begin{definition}[QU-resolution
calculus~\cite{DBLP:conf/cp/Gelder12}]\label{def_qu_resolution} Let $\qclauset
= \prefix.\clauset$ be a PCNF. The rules of the \emph{QU-resolution calculus}
include all rules of \qrescalc (Definition~\ref{def_qres_calculus}) and
additionally the following clause resolution rule.

\begin{align}\tag{$\mathit{qures}$} \AxiomC{$C_1 \cup \{p\}$} \AxiomC{$C_2
\cup \{\bar p\}$} \BinaryInfC{$C_1 \cup C_2$}
\label{rule_qures} \DisplayProof \quad
\begin{minipage}{0.50\textwidth} if for all $x \in \prefix\colon \{x, \bar x\}
\not \subseteq (C_1 \cup C_2)$, \\ $\bar p \not \in C_1$, $p \not \in C_2$,
\\$C_1$,$C_2$ are clauses, and $\quant{\prefix}{p} = \forall$
\end{minipage}
\end{align}

\end{definition}

QU-resolution extends traditional Q-resolution by additionally allowing clause
resolvents over universally quantified pivot variables.

In the following, we use the terminology of \emph{proof systems} (of QBFs) and
\emph{p-simulation} as defined by Cook and
Reckhow~\cite{DBLP:journals/jsyml/CookR79} in an informal way and focus on
unsatisfiable QBFs.

Note that the Q-resolution systems with the generalized and abstraction-based
axioms we defined in the paper are proof systems in the sense of Cook and
Reckhow~\cite{DBLP:journals/jsyml/CookR79} \emph{provided that} the QBF
decision procedures that are used to check the satisfiability of the QBF
$\qclauset[A]$ with respect to a (QCDCL) assignment $A$ produce a proof of
$\qclauset[A]$. These proofs of $\qclauset[A]$ appear as subproofs in the
final proof of $\qclauset$. See also the related remark at the end of
Section~\ref{sec:gen:axioms}.

Informally, a proof system for unsatisfiable QBFs is a function computable in
deterministic polynomial time which associates a proof to every unsatisfiable
QBF $\qclauset$. In most general terms, a proof is a string over some finite
alphabet. We consider \emph{clause resolution proofs} of unsatisfiability
as defined in Section~\ref{sec:qcdcl}.

A proof system $S_2$ \emph{p-simulates} another proof system $S_1$ if there
exists a function computable in deterministic polynomial time which maps any
proof $P_1$ in $S_1$ to a proof $P_2$ in $S_2$. Note that the length of $P_2$
is polynomial in the length of $P_1$.

\begin{theorem} \qrescalcabs p-simulates QU-resolution.
\end{theorem}
\begin{proof}[sketch] Given an unsatisfiable QBF $\qclauset$ and a
QU-resolution proof $P$ of $\qclauset$. We outline how to obtain a
\qrescalcabs proof $P'$ of $\qclauset$ in deterministic polynomial time with
respect to the length of $P$. The length of $P'$ is polynomial in the length
of $P$.  To this end, we identify the inference rules that may appear in the
QU-resolution proof $P$ as proof steps and construct corresponding steps by
inference rules in \qrescalcabs.

The rules~\ref{rule_cl_init}, \ref{rule_cu_init}, \ref{rule_red},
and~\ref{rule_res} are all part of \qrescalc
(Definition~\ref{def_qres_calculus}) and hence by
Definition~\ref{def_qres_calc_with_abs_axiom} are also part of
\qrescalcabs. Hence any derivation by~\ref{rule_cl_init}, \ref{rule_cu_init},
\ref{rule_red}, and~\ref{rule_res} in $P$ can directly be mapped to a
corresponding derivation in a \qrescalcabs proof $P'$.

Consider a derivation of a resolvent $(C_1 \cup C_2)$ from clauses $(C_1 \cup
\{p\})$ and $(C_2 \cup \{\bar p\})$ by rule~\ref{rule_qures} in $P$. Note that
the pivot variable $p$ is universally quantified. Consider the existential
abstraction $\eabs{\qclauset}$ of the given QBF $\qclauset$ and the assignment
$A := \{\bar l \mid l \in (C_1 \cup C_2)\}$ obtained by negating the literals
in the resolvent. Note that $A$ does not contain complementary literals
because the resolvent $(C_1 \cup C_2)$ is nontautological. The formula $\eabs{\qclauset}[A]$ contains the complementary
unit clauses $(p) = (C_1 \cup \{p\})[A]$ and $(\bar p) = (C_2 \cup \{\bar
p\})[A]$. Hence $\eabs{\qclauset}[A]$ is unsatisfiable, which can be
determined in deterministic polynomial time by scanning over the clauses in
$\eabs{\qclauset}[A]$. Therefore, the resolvent $(\bigvee_{l \in A} \bar l) =
(C_1 \cup C_2)$ can be derived by the abstraction-based
axiom~\ref{rule_abs_cl_init} in $P'$.

We have shown how to map derivations of clauses in the QU-resolution proof $P$
to derivations in a \qrescalcabs proof $P'$. The mapping can be computed in
deterministic polynomial time with respect to the size of $P$ because in the
procedure outlined above it is not necessary to invoke a SAT solver (i.e. an
NP oracle) for applications of the abstraction-based
axiom~\ref{rule_abs_cl_init}. \qed
\end{proof}

\section{Additional Experimental Data}

\begin{table}
\begin{center}
\begin{tabular}{l@{\enskip\enskip}c@{\enskip}c@{\enskip}c@{\enskip}c}
\hline
\emph{Solver} & \emph{\#T} & \emph{\#U} & \emph{\#S} & \emph{Time} \\
\hline
\qellnc & 483 & 306 & 177 & 480K \\
\depqbfDynTFTT & 483 & 260 & 223 & 509K \\
\depqbfOnlyDynTF & 481 & 262 & 219 & 528K \\
\depqbfDynBloqqerTFTT & 480 & 257 & 223 & 516K \\
\rareqs & 471 & 272 & 199 & 509K \\
\caqe & 465 & 248 & 217 & 534K \\
\depqbfOnlyDynTT & 464 & 243 & 221 & 526K \\
\depqbfOnlyDynQBCE & 456 & 242 & 214 & 542K \\
\depqbfOnlyDynBloqqer & 450 & 245 & 205 & 550K \\
\qesto & 401 & 212 & 189 & 662K \\
\ghostq & 306 & 148 & 158 & 823K \\
\hline
\end{tabular}
\end{center}
\caption{QBF Gallery 2014 applications track with partial preprocessing by \bloqqer (only QBCE and
  expansion of universal variables), 735 instances. Same column headers as Table~\ref{tab:prepro:track}.}
\label{tab:applications:partial:bloqqer}
\end{table}


\begin{table}
\begin{center}
    \begin{tabular}{l@{\enskip\enskip}c@{\enskip}c@{\enskip}c@{\enskip}c}
      \hline
\emph{Solver} & \emph{\#T} & \emph{\#U} & \emph{\#S} & \emph{Time} \\
      \hline
      \rareqs & 547 & 314 & 233 & 379K \\
      \qellnc & 501 & 301 & 200 & 445K \\
      \qesto & 463 & 248 & 215 & 558K \\
      \depqbfDynTFTT & 434 & 209 & 225 & 579K \\
      \depqbfDynBloqqerTFTT & 432 & 209 & 223 & 585K \\
      \depqbfOnlyDynTT & 426 & 200 & 226 & 586K \\
      \depqbfOnlyDynTF & 418 & 207 & 211 & 623K \\
      \depqbfOnlyDynBloqqer & 409 & 201 & 208 & 622K \\
      \depqbfOnlyDynQBCE & 407 & 200 & 207 & 623K \\
      \caqe & 401 & 193 & 208 & 640K \\
      \ghostq & 350 & 176 & 174 & 739K \\
      \hline
    \end{tabular}
\end{center}
\caption{QBF Gallery 2014 applications track with full preprocessing by \bloqqer, 735 instances. Same column headers as Table~\ref{tab:prepro:track}.}
\label{tab:applications:full:bloqqer}
\end{table}


\begin{table}[ht]
\begin{center}
\begin{tabular}{l@{\enskip\enskip}r@{\enskip}r@{\enskip}r@{\enskip}r}
\hline
\emph{Solver} & \multicolumn{1}{c}{\emph{\#T}} & \multicolumn{1}{c}{\emph{\#U}} & \multicolumn{1}{c}{\emph{\#S}} & \multicolumn{1}{c}{\emph{Time}} \\
\hline
\rareqs & 134 & 68 & 66 & 270K \\
\depqbfOnlyDynTT & 134 & 65 & 69 & 274K \\
\depqbfDynTFTT & 133 & 66 & 67 & 272K \\
\qesto & 133 & 68 & 65 & 274K \\
\depqbfDynBloqqerTFTT & 132 & 66 & 66 & 275K \\
\depqbfOnlyDynQBCE & 132 & 64 & 68 & 277K \\
\depqbfNoQBCETT & 131 & 62 & 69 & 278K \\
\depqbfOnlyDynBloqqer & 130 & 65 & 65 & 278K \\
\depqbfOnlyDynTF & 129 & 64 & 65 & 281K \\
\depqbfNoQBCE & 128 & 61 & 67 & 282K \\
\qellc & 116 & 56 & 60 & 303K \\
\ghostq & 92 & 37 & 55 & 349K \\
\caqe & 89 & 36 & 53 & 347K \\
\hline
\end{tabular}
\end{center}
\caption{QBF Gallery 2014 QBFLIB track (preprocessed with \bloqqer, 276
  instances). Same column headers as Table~\ref{tab:prepro:track}.}
\label{tab:qbflib:track:bloqqer}
\end{table}

\end{appendix}


\begin{thebibliography}{10}
\providecommand{\url}[1]{\texttt{#1}}
\providecommand{\urlprefix}{URL }

\bibitem{DBLP:conf/fmcad/AyariB02}
Ayari, A., Basin, D.A.: {QUBOS: Deciding Quantified Boolean Logic Using
  Propositional Satisfiability Solvers}. In: {FMCAD}. LNCS, vol. 2517, pp.
  187--201. Springer (2002)

\bibitem{DBLP:conf/sat/BalabanovWJ14}
Balabanov, V., Widl, M., Jiang, J.R.: {QBF Resolution Systems and Their Proof
  Complexities}. In: {SAT}. LNCS, vol. 8561, pp. 154--169. Springer (2014)

\bibitem{DBLP:journals/jsat/BenedettiM08}
Benedetti, M., Mangassarian, H.: {QBF-Based Formal Verification: Experience and
  Perspectives}. {JSAT}  5(1-4),  133--191 (2008)

\bibitem{DBLP:conf/mfcs/BeyersdorffCJ14}
Beyersdorff, O., Chew, L., Janota, M.: {On Unification of {QBF}
  Resolution-Based Calculi}. In: {MFCS}. LNCS, vol. 8635, pp. 81--93. Springer
  (2014)

\bibitem{beyersdorff_et_al:LIPIcs:2015:4905}
Beyersdorff, O., Chew, L., Janota, M.: {Proof Complexity of Resolution-based
  QBF Calculi}. In: STACS. Leibniz International Proceedings in Informatics
  (LIPIcs), vol.~30, pp. 76--89. Schloss Dagstuhl--Leibniz-Zentrum fuer
  Informatik (2015)

\bibitem{DBLP:conf/sat/Biere04a}
Biere, A.: {Resolve and Expand}. In: SAT. LNCS, vol. 3542, pp. 59--70. Springer
  (2004)

\bibitem{DBLP:journals/jsat/Biere08}
Biere, A.: {PicoSAT Essentials}. {JSAT}  4(2-4),  75--97 (2008)

\bibitem{DBLP:conf/cade/BiereLS11}
Biere, A., Lonsing, F., Seidl, M.: {Blocked Clause Elimination for {QBF}}. In:
  CADE. LNCS, vol. 6803, pp. 101--115. Springer (2011)

\bibitem{DBLP:conf/aaai/CadoliGS98}
Cadoli, M., Giovanardi, A., Schaerf, M.: {An Algorithm to Evaluate Quantified
  Boolean Formulae}. In: {AAAI}. pp. 262--267. {AAAI} Press / The {MIT} Press
  (1998)

\bibitem{DBLP:journals/jsyml/CookR79}
Cook, S.A., Reckhow, R.A.: {The Relative Efficiency of Propositional Proof
  Systems}. J. Symb. Log.  44(1),  36--50 (1979)

\bibitem{DBLP:series/faia/GiunchigliaMN09}
Giunchiglia, E., Marin, P., Narizzano, M.: {Reasoning with Quantified Boolean
  Formulas}. In: Handbook of Satisfiability, FAIA, vol. 185, pp. 761--780.
  {IOS} Press (2009)

\bibitem{DBLP:journals/jair/GiunchigliaNT06}
Giunchiglia, E., Narizzano, M., Tacchella, A.: {Clause/Term Resolution and
  Learning in the Evaluation of Quantified Boolean Formulas}. JAIR  26,
  371--416 (2006)

\bibitem{DBLP:journals/jair/HeuleJLSB15}
Heule, M., J{\"{a}}rvisalo, M., Lonsing, F., Seidl, M., Biere, A.: {Clause
  Elimination for {SAT} and {QSAT}}. JAIR  53,  127--168 (2015)

\bibitem{DBLP:conf/cade/HeuleSB14}
Heule, M., Seidl, M., Biere, A.: {A Unified Proof System for {QBF}
  Preprocessing}. In: {IJCAR}. LNCS, vol. 8562, pp. 91--106. Springer (2014)

\bibitem{gallery14}
Janota, M., Jordan, C., Klieber, W., Lonsing, F., Seidl, M., Van~Gelder, A.:
  {The QBF Gallery 2014: The QBF Competition at the FLoC Olympic Games}. JSAT
  9,  187--206 (2015)

\bibitem{DBLP:journals/tcs/JanotaM15}
Janota, M., Marques{-}Silva, J.: {Expansion-based QBF solving versus
  Q-resolution}. Theor. Comput. Sci.  577,  25--42 (2015)

\bibitem{DBLP:conf/ijcai/JanotaM15}
Janota, M., Marques{-}Silva, J.: {Solving {QBF} by Clause Selection}. In:
  {IJCAI}. pp. 325--331. {AAAI} Press (2015)

\bibitem{Janota20161}
Janota, M., Klieber, W., Marques-Silva, J., Clarke, E.: {{Solving QBF with
  counterexample guided refinement}}. Artif.~Intell.  234,  1--25 (2016)

\bibitem{DBLP:journals/iandc/BuningKF95}
Kleine{ }B{\"{u}}ning, H., Karpinski, M., Fl{\"{o}}gel, A.: {Resolution for
  Quantified Boolean Formulas}. Inf. Comput.  117(1),  12--18 (1995)

\bibitem{DBLP:conf/sat/KlieberSGC10}
Klieber, W., Sapra, S., Gao, S., Clarke, E.M.: {A Non-prenex, Non-clausal {QBF}
  Solver with Game-State Learning}. In: SAT. LNCS, vol. 6175, pp. 128--142.
  Springer (2010)

\bibitem{DBLP:conf/tableaux/Letz02}
Letz, R.: {Lemma and Model Caching in Decision Procedures for Quantified
  Boolean Formulas}. In: TABLEAUX. LNCS, vol. 2381, pp. 160--175. Springer
  (2002)

\bibitem{DBLP:conf/lpar/LonsingBBES15}
Lonsing, F., Bacchus, F., Biere, A., Egly, U., Seidl, M.: {Enhancing
  Search-Based {QBF} Solving by Dynamic Blocked Clause Elimination}. In: LPAR.
  LNCS, vol. 9450, pp. 418--433. Springer (2015)

\bibitem{DBLP:conf/sat/LonsingB11}
Lonsing, F., Biere, A.: {Failed Literal Detection for QBF}. In: SAT. LNCS, vol.
  6695, pp. 259--272. Springer (2011)

\bibitem{Lonsing201692}
Lonsing, F., Seidl, M., Van~Gelder, A.: {The {QBF} Gallery: Behind the scenes}.
  Artif.~Intell.  237,  92--114 (2016)

\bibitem{DBLP:conf/fmcad/RabeT15}
Rabe, M.N., Tentrup, L.: {CAQE: A Certifying QBF Solver}. In: {FMCAD}. pp.
  136--143. {IEEE} (2015)

\bibitem{DBLP:conf/mbmv/ReimerPSB12}
Reimer, S., Pigorsch, F., Scholl, C., Becker, B.: {Enhanced Integration of QBF
  Solving Techniques}. In: Methoden und Beschreibungssprachen zur Modellierung
  und Verifikation von Schaltungen und Systemen (MBMV). pp. 133--143. Verlag
  Dr. Kovac (2012)

\bibitem{DBLP:journals/jar/SamerS09}
Samer, M., Szeider, S.: {Backdoor Sets of Quantified Boolean Formulas}. JAR
  42(1),  77--97 (2009)

\bibitem{DBLP:conf/cp/SamulowitzB05}
Samulowitz, H., Bacchus, F.: {Using SAT in QBF}. In: {CP}. LNCS, vol. 3709, pp.
  578--592. Springer (2005)

\bibitem{DBLP:conf/cp/SamulowitzDB06}
Samulowitz, H., Davies, J., Bacchus, F.: {Preprocessing QBF}. In: CP. LNCS,
  vol. 4204, pp. 514--529. Springer (2006)

\bibitem{DBLP:series/faia/SilvaLM09}
Silva, J.P.M., Lynce, I., Malik, S.: {Conflict-Driven Clause Learning {SAT}
  Solvers}. In: Handbook of Satisfiability, FAIA, vol. 185, pp. 131--153. {IOS}
  Press (2009)

\bibitem{DBLP:conf/sat/TuHJ15}
Tu, K., Hsu, T., Jiang, J.R.: {QELL: QBF Reasoning with Extended Clause
  Learning and Levelized SAT Solving}. In: SAT. LNCS, vol. 9340, pp. 343--359.
  Springer (2015)

\bibitem{DBLP:conf/cp/Gelder12}
Van~Gelder, A.: {Contributions to the Theory of Practical Quantified Boolean
  Formula Solving}. In: CP. LNCS, vol. 7514, pp. 647--663. Springer (2012)

\bibitem{DBLP:conf/iccad/ZhangM02}
Zhang, L., Malik, S.: {Conflict Driven Learning in a Quantified Boolean
  Satisfiability Solver}. In: {ICCAD}. pp. 442--449. {ACM} / {IEEE} Computer
  Society (2002)

\bibitem{DBLP:conf/cp/ZhangM02}
Zhang, L., Malik, S.: {Towards a Symmetric Treatment of Satisfaction and
  Conflicts in Quantified Boolean Formula Evaluation}. In: CP. LNCS, vol. 2470,
  pp. 200--215. Springer (2002)

\end{thebibliography}
\end{document}